\documentclass[a4paper, 10pt]{article}

\usepackage{amsfonts, amssymb}
\usepackage{amscd}
\usepackage{amsmath}
\usepackage{comment}
\usepackage{pdflscape}
\usepackage{amsthm}
\usepackage{color}

\newtheorem{theo}{Theorem}[section]
\newtheorem{defi}[theo]{Definition}
\newtheorem{coro}[theo]{Corollary}
\newtheorem{prop}[theo]{Proposition}
\newtheorem{lemm}[theo]{Lemma}

\newcommand{\ZZ}{{\mathbb{Z}}}
\newcommand{\NN}{{\mathbb{N}}}

\newcommand{\emm}{\mathfrak{m}}

\title{$\ZZ_2$-double cyclic codes\thanks{%
  This work has been partially supported by the Spanish MICINN grant TIN2013-40524-P and by the Catalan grant 2009SGR1224.}}

\author{ Joaquim Borges Ayats\\Department of Information and Communication Engineering \\
Universitat Aut\`onoma de Barcelona \\ 08193-Bellaterra, Spain \\ Cristina Fern\'andez-C\'ordoba  \\Department of Information and Communication Engineering \\
Universitat Aut\`onoma de Barcelona \\ 08193-Bellaterra, Spain \\
and  \\ Roger Ten-Valls  \\Department of Information and
Communication Engineering \\
Universitat Aut\`onoma de Barcelona \\ 08193-Bellaterra, Spain
}

\date{\today}

\begin{document}

\maketitle

\begin{abstract}
A binary linear code $C$ is a $\mathbb{Z}_2$-double cyclic code if the set of coordinates can be partitioned into two subsets such that any cyclic shift of the coordinates of both subsets leaves invariant the code. These codes can be identified as submodules of the $\mathbb{Z}_2[x]$-module $\mathbb{Z}_2[x]/(x^r-1)\times\mathbb{Z}_2[x]/(x^s-1).$ We determine the structure of $\mathbb{Z}_2$-double cyclic codes giving the generator polynomials of these codes. The related polynomial representation of $\mathbb{Z}_2$-double cyclic codes and its duals, and the relations between the polynomial generators of these codes are studied.
\end{abstract}

\section{Introduction}

Let $\mathbb{Z}_2$ be the ring of integers modulo $2$. Let $\mathbb{Z}_2^n$ denote the set of all binary vectors of length $n$. Any non-empty subset of $\mathbb{Z}_2^n$ is a binary code and a subgroup of $\mathbb{Z}_2^n$ is called a \textit{binary linear code}. In this paper we introduce a subfamily of binary linear codes, called \textit{$\mathbb{Z}_2$-double cyclic codes}, with the property that the set of coordinates can be partitioned into two subsets, the first $r$ coordinates and the last $s$ coordinates, such that any cyclic shift of the coordinates of both subsets of a codeword is also a codeword.

Notice that if one of these sets of coordinates is empty, for example $r=0$, then we obtain a binary cyclic code of length $s$. So, binary cyclic codes are a special class of $\mathbb{Z}_2$-double cyclic codes. Most of the theory of binary cyclic codes can be found in \cite{macwilliams}. Another special case is when $r=s$, where a $\mathbb{Z}_2$-double cyclic code is permutation equivalent to a quasi-cyclic code of index 2 and even length (see \cite{macwilliams}).

In recent times, $\mathbb{Z}_2\mathbb{Z}_4$-additive codes have been studied (see \cite{AddDual}, \cite{Z2Z4RK}). For $\mathbb{Z}_2\mathbb{Z}_4$-additive codes, the set of coordinates is partitioned into two subsets, the first one of binary coordinates and the second one of quaternary coordinates. The simultaneous cyclic shift of the subsets of coordinates of a codeword has been first defined in  \cite{Abu}, that studies $\mathbb{Z}_2\mathbb{Z}_4$-additive cyclic codes and these codes can be identified as $\mathbb{Z}_4[x]$-modules of a certain ring. The duality of $\ZZ_2\ZZ_4$-additive cyclic codes is studied in \cite{Z2Z4CDual}.

The aim of this paper is the study of the algebraic structure of $\mathbb{Z}_2$-double cyclic codes and their dual codes. It is organized as follows. In Section \ref{sec:def}, we give the definition of $\mathbb{Z}_2$-double cyclic codes, we find the relation between some canonical projections of these codes and binary cyclic codes and we present the $\mathbb{Z}_2[x]$-module $\mathbb{Z}_2[x]/(x^r-1)\times\mathbb{Z}_2[x]/(x^s-1)$, denoted by $R_{r,s}$. In Section \ref{sec:AlgStru&Gen}, we discuss about the algebraic structure of a $\mathbb{Z}_2$-double cyclic code and we state some relations between its generators. In Section \ref{sec:Duality}, we study the concept of duality and, given a $\mathbb{Z}_2$-double cyclic code, we determine the generators of the dual code in terms of the generators of the code.

\section{$\ZZ_2$-double cyclic codes}\label{sec:def}
Let $C$ be a binary code of length $n$. Let $r$ and $s$ be integers such that $n=r+s$. We consider a partition of the set of the $n$ coordinates into two subsets of $r$ and $s$ coordinates, respectively, so that $C$ is a subset of $\mathbb{Z}_2^{r}\times\mathbb{Z}_2^{s}$.

\begin{defi}
Let $C$ be a binary linear code of length $n=r+s$. The code $C$ is called \emph{$\mathbb{Z}_2$-double cyclic} if 
$$ (u_0, u_1,\dots, u_{r-2}, u_{r-1}\mid  u'_0,u'_1,\dots, u'_{s-2}, u'_{s-1})\in C $$ 
implies 
$$(u_{r-1}, u_0, u_1,\dots, u_{r-2}\mid u'_{s-1},u'_0,u'_1,\dots, u'_{s-2})\in C.$$
\end{defi}

Let ${\bf u}=(u_{0},u_{1},\dots,u_{r-1}\mid u'_{0},\dots,u'_{s-1})$ be a codeword in $C$ and $i$ be an integer, then we denote by $${\bf u}^{(i)}=(u_{0+i},u_{1+i},\dots,u_{r-1+i}\mid u'_{0+i},\dots,u'_{s-1+i})$$ 
the $i$th shift of ${\bf u}$, where the subscripts are read modulo $r$ and $s$, respectively.

Let $C_r$ be the canonical projection of $C$ on the first $r$ coordinates and $C_s$ on the last $s$ coordinates. The canonical projection is a linear map. Then, $C_r$ and $C_s$ are binary cyclic codes of length $r$ and $s$, respectively. A code $C$ is called \textit{separable} if $C$ is the direct product of $C_r$ and $C_s$.

There is a bijective map between $\mathbb{Z}_2^r\times \mathbb{Z}_2^s$ and $\mathbb{Z}_2[x]/(x^r-1)\times\mathbb{Z}_2[x]/(x^s-1)$ given by:
{\small $$(u_0, u_1,\dots, u_{r-1}\mid  u'_0,\dots, u'_{s-1})\mapsto (u_0+ u_1x+\dots +  u_{r-1}x^{r-1}\mid  u'_0+\dots +u'_{s-1}x^{s-1}).$$}

We denote the image of the vector ${\bf u}$ by ${\bf u}(x).$
\begin{defi}
Denote by $R_{r,s}$ the ring $\mathbb{Z}_2[x]/(x^r-1)\times\mathbb{Z}_2[x]/(x^s-1)$. We define the operation
$$\star : \mathbb{Z}_2[x]\times R_{r,s}\rightarrow R_{r,s}$$
as
$$\lambda(x)\star (p(x)\mid  q(x))= (\lambda(x)p(x) \mid  \lambda(x)q(x)),$$
where $\lambda(x)\in\mathbb{Z}_2[x]$ and $(p(x)\mid  q(x))\in R_{r,s}$. 
\end{defi}

The ring $R_{r,s}$ with the external operation $\star$ is a $\mathbb{Z}_2[x]$-module. Let ${\bf u}(x)=(u(x)\mid u'(x))$ be an element of $ R_{r,s}$. Note that if we operate ${\bf u}(x)$ by $x$ we get
\begin{align*}
x\star{\bf u}(x)&= x\star(u(x)\mid u'(x))\\
&= (u_0x+\dots+  u_{r-2}x^{r-1} +  u_{r-1}x^{r}\mid  u'_0x+\dots +  u'_{s-2}x^{s-1}+u'_{s-1}x^{s})\\
&=( u_{r-1} + u_0x+\dots+  u_{r-2}x^{r-1} \mid  u'_{s-1} + u'_0x+\dots +  u'_{s-2}x^{s-1}).
\end{align*}
Hence, $x\star{\bf u}(x)$ is the image of the vector ${\bf u}^{(1)}$. Thus, the operation of ${\bf u}(x)$ by $x$ in $R_{r,s}$ corresponds to a shift of ${\bf u}$. In general, $x^i\star {\bf u}(x)={\bf u}^{(i)}(x)$ for all $i$.

\section{Algebraic structure and generators}\label{sec:AlgStru&Gen}

In this section, we study submodules of $R_{r,s}$. We describe the generators of such submodules and state some properties. From now on, $\langle S\rangle$ will denote the submodule generated by a subset $S$ of $R_{r,s}.$

\begin{theo}\label{Z2submodulesform}
The $\mathbb{Z}_2[x]$-module $R_{r,s}$ is a noetherian $\mathbb{Z}_2[x]$-module, and every submodule $C$ of $R_{r,s}$ can be written as
$$C=\langle(b(x)\mid { 0}),(\ell(x)\mid  a(x))\rangle,$$
where $b(x), \ell(x)\in\mathbb{Z}_2[x]/(x^r-1)$ with $b(x)\mid (x^r-1)$ and $a(x)\in \mathbb{Z}_2[x]/(x^s-1)$ with $a(x)\mid (x^s-1)$.
\end{theo}

\begin{proof}
Let $\pi_r : R_{r,s} \rightarrow \ZZ_2[x]/(x^r-1)$ and $\pi_s : R_{r,s} \rightarrow \ZZ_2[x]/(x^s-1)$ be the canonical projections, let $C$ be a submodule of $R_{r,s}$.\\
As $\ZZ_2[x]/(x^s-1)$ is noetherian then $C_s= \pi_s(C)$ is finitely generated.\\
Define $C'=\{ (p(x)|q(x))\in C\mid q(x)=0\}$. It is easy to check that $C'\cong \pi_r(C')$ by $(p(x)\mid { 0})\mapsto p(x)$. Hence $\ZZ_2[x]/(x^r-1)$ is noetherian, $\pi_r(C')$ is finitely generated and so is $C'$.\\
Let $b(x)$ be a generator of $\pi_r(C')$, then $b(x)\mid (x^r-1)$ and $(b(x)\mid { 0})$ is a generator of $C'$.
Let $a(x)\in C_s$ such that $C_s = \langle a(x)\rangle,$ then $a(x)\mid (x^s-1)$ and there exists $\ell(x)\in\ZZ_2[x]/(x^r-1)$ such that $(\ell(x)\mid a(x))\in C.$\\
We claim that 
$$C=\langle(b(x)\mid { 0}),(\ell(x)\mid  a(x))\rangle.$$
Let $(p(x)\mid q(x))\in C$, then $q(x)= \pi_s((p(x)\mid q(x)))\in C_s.$ So, there exists $\lambda(x)\in \ZZ_2[x]$ such that $q(x)=  \lambda(x)a(x)$. Now,
$$(p(x)\mid q(x)) - \lambda(x)\star(\ell(x)\mid a(x)) = (p(x)-\lambda(x)\ell(x)\mid { 0})\in C'.$$
Then there exists $\mu(x)\in \ZZ_2[x]$ such that $(p(x)-\lambda(x)\ell(x)\mid { 0})= \mu(x)\star(b(x)\mid { 0}).$ Thus,
$$(p(x)\mid q(x))= \mu(x)\star(b(x)\mid { 0}) + \lambda(x)\star(\ell(x)\mid a(x)).$$
So, $C$ is finitely generated by $\langle(b(x)\mid { 0}),(\ell(x)\mid  a(x))\rangle$ and then $R_{r,s}$ is a noetherian $\ZZ_2[x]$-module.
\end{proof}

From the previous results, it is clear that we can identify $\mathbb{Z}_2$-double cyclic codes in $\mathbb{Z}_2^{r}\times\mathbb{Z}_2^{s}$ as submodules of $R_{r,s}$. So, any submodule of $R_{r,s}$ is a $\mathbb{Z}_2$-double cyclic code.

Note that if $C$ is a $\mathbb{Z}_2$-double cyclic code with $C=\langle(b(x)\mid { 0}),(\ell(x)\mid  a(x))\rangle$, then the canonical projections $C_r$ and $C_s$ are binary cyclic codes generated by $\gcd(b(x),\ell(x))$ and $a(x)$, respectively.

On the one hand, we have seen that $R_{r,s}$ is a $\ZZ_2[x]$-module, and multiply by $x\in\ZZ_2[x]$ is the right shift on the vector space $\ZZ_2^r\times \ZZ_2^s$. On the other hand, we have that $\ZZ_2^r\times \ZZ_2^s$ is a $\ZZ_2$-module, where the operations are addition and multiplication by elements of $\ZZ_2$.

So, our goal now is to find a set of generators for $C$ as a $\ZZ_2$-module. We will denote the $\ZZ_2$-linear combinations of elements of a subset $S\subseteq R_{r,s}$ by $\langle S\rangle_{\ZZ_2}=\{\sum_i \lambda_i s_i \mid \lambda_i\in \ZZ_2, s_i\in S\}$, and we will call a set $S$ a \emph{$\ZZ_2$-linear independent} set if the relation 
$\sum_i \lambda_is_i=0$
implies that $\lambda_is_i=0$ for all $i$.

\begin{prop}\label{spaningset}
Let $C=\langle(b(x)\mid { 0}),(\ell(x)\mid  a(x))\rangle$ be a $\mathbb{Z}_2$-double cyclic code. Define the sets
$$S_1=\{(b(x)\mid 0), x\star(b(x)\mid 0),\dots, x^{r-\deg(b(x))-1}\star(b(x)\mid 0)\},$$
$$S_2=\{(\ell(x)\mid a(x)), x\star(\ell(x)\mid a(x)),\dots, x^{s-\deg(a(x))-1}\star(\ell(x)\mid a(x))\}. $$
Then, $S_1\cup S_2$ forms a generating set for $C$ as a $\ZZ_2$-module.
\end{prop}
\begin{proof}
It is easy to check that the codewords of $S_1\cup S_2$ are $\ZZ_2$-linear independent.

Let $c(x)\in C$, such that $c(x)=p_1(x)\star(b(x)\mid 0) + p_2(x)\star(\ell(x)\mid  a(x))$.
We have to check that $c(x)\in\langle S_1\cup S_2\rangle_{\ZZ_2}.$

If $\deg(p_1(x))< r-\deg(b(x))-1,$ then $p_1(x)\star(b(x)\mid 0)\in \langle S_1\rangle_{\ZZ_2}$. Otherwise, using the division algorithm, we compute $p_1(x)=q_1(x)\frac{x^r-1}{b(x)}+r_1(x)$ with $\deg(r_1(x))< r-\deg(b(x))-1,$ so
$$p_1(x)\star(b(x)\mid 0)=\left( q_1(x)\frac{x^r-1}{b(x)}+r_1(x)\right)\star(b(x)\mid 0)=r_1(x)\star(b(x)\mid 0)\in\langle S_1\rangle_{\ZZ_2}.$$

So, $c(x)\in \langle S_1\cup S_2\rangle_{\ZZ_2}$ if $p_2(x)\star(\ell(x)\mid  a(x))\in \langle S_1\cup S_2\rangle_{\ZZ_2}.$

If $\deg(p_2(x))< s-\deg(a(x))-1,$ then $p_2(x)\star(\ell(x)\mid a(x))\in \langle S_2\rangle_{\ZZ_2}$. If not, using the division algorithm, consider 
$p_2(x)=q_2(x)\frac{x^s-1}{a(x)}+r_2(x)$ where $\deg(r_2(x))< s-\deg(a(x))-1.$ Then,
\begin{align*}
p_2(x)\star(\ell(x)\mid  a(x))&=\left( q_2(x)\frac{x^s-1}{a(x)}+r_2(x)\right)\star(\ell(x)\mid  a(x))\\
&=\left( q_2(x)\frac{x^s-1}{a(x)}\right)\star(\ell(x)\mid  a(x)) + r_2(x)\star(\ell(x)\mid  a(x)).
\end{align*}
On the one hand, $r_2(x)\star(\ell(x)\mid  a(x))\in\langle S_2\rangle_{\ZZ_2}.$ On the other hand, 
$$\left( q_2(x)\frac{x^s-1}{a(x)}\right)\star(\ell(x)\mid  a(x))=(q_2(x)\frac{x^s-1}{a(x)}\ell(x)\mid  0).$$
By Proposition \ref{BdiviLXS/A}, $b(x)$ divides $\frac{x^s-1}{a(x)}\ell(x)$ and it follows straightforward that $(q_2(x)\frac{x^s-1}{a(x)}\ell(x)\mid  0)\in \langle S_1\rangle_{\ZZ_2}$. Thus, $c(x)\in \langle S_1\cup S_2\rangle_{\ZZ_2}.$
\end{proof}

\begin{prop}
Let $C=\langle(b(x)\mid { 0}),(\ell(x)\mid  a(x))\rangle$ be a $\mathbb{Z}_2$-double cyclic code. Then, $C$ is permutation equivalent to a binary linear code with generator matrix of the form
$$G =
 \left(\begin{array}{ccc|ccc}
  I_{r-\deg(b(x))} & A_1 & A_2 & 0 & 0 & 0 \\
  0 & B_{\kappa} & B & C_1 & I_\kappa &  0 \\
  0 & 0 & 0 & C_2  & R & I_{s-\deg(a(x))-\kappa} 
 \end{array}\right),
$$
where $B_\kappa$ is a square matrix of full rank and $\kappa=\deg(b(x))-\deg(\gcd(b(x),\ell(x)))$.
\end{prop}
\begin{proof}

Let $C$ be a $\mathbb{Z}_2$-double cyclic code with $C=\langle(b(x)\mid { 0}),(\ell(x)\mid  a(x))\rangle$. Then by Proposition \ref{spaningset}, $C$ is generated by the matrix whose rows are the elements of the set $S_1\cup S_2.$

Since $r-\deg(b(x))$ and $s-\deg(a(x))$ are the dimensions of the matrices generated by the shifts of $b(x)$ and $a(x)$, respectively, the code $C$ is permutation equivalent to a code with generator matrix of the form
$$
\left(\begin{array}{cc|cc}
  I_{r-\deg(b(x))} & A' & 0 & 0 \\
  0 & B' & C'  & I_{s-\deg(a(x))} 
 \end{array}\right).
$$

It is known that $C_r$ is a linear cyclic code generated by $\gcd(b(x), \ell(x))$, then the submatrix $B'$ has rank $\kappa=\deg(b(x))-\deg(\gcd(b(x),\ell(x))).$ Moreover, $C_r$ is permutation equivalent to a linear code generated by the matrix 
$$
\left(\begin{array}{ccc}
  I_{r-\deg(b(x))} & A_1 & A_2  \\
  0 & B_\kappa & B\\ 
  0 & 0 & 0
 \end{array}\right),
$$
with $B_\kappa$ a full rank square matrix of size $\kappa\times\kappa.$ Finally, applying the convenient permutations and linear combinations, we have that $C$ is permutation equivalent to a linear code with generator matrix 
$$
 \left(\begin{array}{ccc|ccc}
  I_{r-\deg(b(x))} & A_1 & A_2 & 0 & 0 & 0 \\
  0 & B_{\kappa} & B & C_1 & I_\kappa &  0 \\
  0 & 0 & 0 & C_2  & R & I_{s-\deg(a(x))-\kappa} 
 \end{array}\right).
$$
\end{proof}
\begin{coro}
Let $C=\langle(b(x)\mid { 0}),(\ell(x)\mid  a(x))\rangle$ be a $\mathbb{Z}_2$-double cyclic code. Then, $C$ is a binary linear code of dimension $r+s-\deg (b(x))-\deg(a(x)).$
\end{coro}

\begin{prop}
Let $C=\langle(b(x)\mid { 0}),(\ell(x)\mid  a(x))\rangle$ be a $\mathbb{Z}_2$-double cyclic code. Then, we can assume that $\deg (\ell(x))<\deg(b(x))$.
\end{prop}
\begin{proof}
Suppose that $\deg (\ell(x))\geq \deg(b(x))$. Let $i=\deg (\ell(x))-\deg(b(x))$ and let $C'$ be the code generated by
$$C'= \langle(b(x)\mid { 0}),(\ell(x)+x^i\star b(x)\mid  a(x))\rangle.$$
On the one hand, $\deg(\ell(x)+x^i\star b(x))<\deg(\ell(x))$ and since the generators of $C'$ belongs to $C$, we have that $C' \subseteq C$. On the other hand,
$$(\ell(x)\mid a(x))= (\ell(x)+x^i\star b(x)\mid  a(x))+x^i\star (b(x)\mid { 0}).$$
Then, $\langle(\ell(x)\mid  a(x))\rangle\subseteq C'$ and hence $C\subseteq C'$.
Thus, $C = C'.$ 
\end{proof}

\begin{prop}\label{BdiviLXS/A}
Let $C=\langle(b(x)\mid { 0}),(\ell(x)\mid  a(x))\rangle$ be a $\mathbb{Z}_2$-double cyclic code. Then, $b(x)\mid\frac{x^s-1}{a(x)}\ell (x).$
\end{prop}
\begin{proof}
Let $\pi$ be the projective homomorphism of $\ZZ_2[x]$-modules defined by:
$$\begin{array}{cccc}
\pi : & C &\longrightarrow & \ZZ_2[x]/(x^s-1) \\
 & (p_1(x)\mid  p_2(x))&\longrightarrow & p_2(x)
\end{array}$$
It can be easily checked that $\ker (\pi) = \langle(b(x)\mid { 0})\rangle.$\\ 
Now, consider $\frac{x^s-1}{a(x)}\star (\ell(x)\mid  a(x)) = (\frac{x^s-1}{a(x)}\ell(x)\mid { 0}).$ 
So, 
$$\frac{x^s-1}{a(x)}\star (\ell(x)\mid  a(x))\in \ker (\pi)=\langle(b(x)\mid { 0})\rangle.$$ 
Thus, $b(x)\mid\frac{x^s-1}{a(x)}\ell(x).$
\end{proof}

\begin{coro}\label{bdiviXSgcd}
Let $C=\langle(b(x)\mid { 0}),(\ell(x)\mid  a(x))\rangle$ be a $\mathbb{Z}_2$-double cyclic code. Then, $b(x)\mid\frac{x^s-1}{a(x)}\gcd(\ell (x), b(x)).$
\end{coro}

\begin{prop}
Let $C=\langle(b(x)\mid { 0}),(\ell(x)\mid  a(x))\rangle$ be a separable $\mathbb{Z}_2$-double cyclic code. Then, $\ell (x)=0.$
\end{prop}

\section{Duality}\label{sec:Duality}

Let $C$ be a $\mathbb{Z}_2$-double cyclic code and $C^\perp$ the dual code of $C$ (see \cite{huffman}).
Taking a vector ${\bf v}$ of $C^\perp$, ${\bf u}\cdot {\bf v}=0$ for all ${\bf u}$ in $C$. Since ${\bf u}$ belongs to $C$, we know that ${\bf u}^{(-1)}$ is also a codeword. So, ${\bf u}^{(-1)}\cdot {\bf v}={\bf u}\cdot{\bf v}^{(1)}=0$ for all ${\bf u}$ from $C$, therefore ${\bf v}^{(1)}$ is in $C^\perp$ and $C^\perp$ is also a $\mathbb{Z}_2$-double cyclic code. Consequently, we obtain the following proposition.

\begin{prop}\label{DualDoubleCyclic}
Let $C$ be a $\mathbb{Z}_2$-double cyclic code. Then the dual code of $C$ is also a $\mathbb{Z}_2$-double cyclic code. We denote 
$$C^\perp = \langle (\bar{b}(x)\mid { 0}), (\bar{\ell}(x) \mid  \bar{a}(x)) \rangle,$$
where $\bar{b}(x), \bar{\ell}(x)\in\mathbb{Z}_2[x]/(x^r-1)$ with $\bar{b}(x)\mid(x^r-1)$ and $\bar{a}(x)\in \mathbb{Z}_2[x]/(x^s-1)$ with $\bar{a}(x)\mid(x^s-1)$.  
\end{prop}

The \textit{reciprocal polynomial} of a polynomial $p(x)$ is $x^{\deg(p(x))}p(x^{-1})$ and is denoted by $p^*(x)$. As in the theory of binary cyclic codes, reciprocal polynomials have an important role in the duality (see \cite{macwilliams}).

We denote the polynomial $\sum^{m-1}_{i=0} x^i$ by $\theta_m(x)$. Using this notation we have the following proposition.

\begin{prop}\label{xnm=xntheta}
Let $n,m\in \NN$. Then, $x^{nm}-1= (x^n-1)\theta_m(x^n).$
\end{prop}
\begin{proof}
It is well know that $y^m-1=(y-1)\theta_m(y)$, replacing $y$ by $x^n$ the result follows.
\end{proof}
From now on, $\emm$ denotes the least common multiple of $r$ and $s$.
\begin{defi}
Let ${\bf u}(x)=(u(x)\mid u'(x))$ and ${\bf v}(x)=(v(x)\mid v'(x))$ be elements in $R_{r,s}$. We define the map
$$\circ:R_{r,s}\times R_{r,s}\longrightarrow \mathbb{Z}_2[x]/(x^{\emm}-1),$$
such that
\begin{align*}
\circ({\bf u}(x),{\bf v}(x))=& u(x)\theta_{\frac{\emm}{r}}(x^r)x^{\emm-1-\deg(v(x))}v^*(x) +\\
&+  u'(x)\theta_\frac{\emm}{s}(x^s)x^{\emm-1-\deg(v'(x))}{v'}^*(x) \mod(x^{\emm}-1).
\end{align*}
 
\end{defi}
The map $\circ$ is linear in each of its arguments; i.e., if we fix the first entry of the map invariant, while letting the second entry vary, then the result is a linear map. Similarly, when fixing the second entry invariant. Then, the map $\circ$ is a bilinear map between $\mathbb{Z}_2[x]$-modules.

From now on, we denote $\circ({\bf u}(x), {\bf v}(x))$ by ${\bf u}(x)\circ{\bf v}(x)$. Note that ${\bf u}(x)\circ{\bf v}(x)$ belongs to $\mathbb{Z}_2[x]/(x^{\emm}-1)$.

\begin{prop}
Let ${\bf u}$ and ${\bf v}$ be vectors in $\mathbb{Z}_2^{r}\times\mathbb{Z}_2^{s}$ with associated polynomials ${\bf u}(x)=(u(x)\mid u'(x))$ and ${\bf v}(x)=(v(x)\mid v'(x))$, respectively. Then, ${\bf u}$ is orthogonal to ${\bf v}$ and all its shifts if and only if $${\bf u}(x)\circ{\bf v}(x)= 0\mod(x^{\emm}-1).$$ 
\end{prop}
\begin{proof}
Let ${\bf v}^{(i)}=(v_{0+i}v_{1+i}\ldots v_{r-1+i}\mid v'_{0+i}\ldots v'_{s-1+i})$ be the $i$th shift of ${\bf v}$. Then, 

$${\bf u}\cdot{\bf v}^{(i)}=0\mbox{ if and only if }\sum^{r-1}_{j=0} u_jv_{j+i} +\sum^{s-1}_{k=0} u'_kv'_{k+i}=0.$$
Let $S_i=\sum^{r-1}_{j=0} u_jv_{j+i} +\sum^{s-1}_{k=0} u'_kv'_{k+i}$. One can check that  

\begin{align*}
{\bf u}(x)\circ{\bf v}(x)&=\sum^{r-1}_{n=0}\left[ \theta_\frac{\emm}{r}(x^r)\sum^{r-1}_{j=0}u_jv_{j+n}x^{\emm-1-n}\right] +\cdots\\ 
&\cdots + \sum^{s-1}_{t=0}\left[  \theta_\frac{\emm}{s}(x^s)\sum^{s-1}_{k=0}u'_kv'_{k+t}x^{\emm-1-t}\right]\\
&=\theta_\frac{\emm}{r}(x^r)\left[\sum^{r-1}_{n=0} \sum^{r-1}_{j=0}u_jv_{j+n}x^{\emm-1-n}\right] +\cdots\\ 
&\cdots + \theta_\frac{\emm}{s}(x^s)\left[\sum^{s-1}_{t=0}  \sum^{s-1}_{k=0}u'_kv'_{k+t}x^{\emm-1-t}\right].
\end{align*}

Then, arranging the terms one obtains that

$${\bf u}(x)\circ{\bf v}(x)=\sum^{\emm-1}_{i=0} S_i x^{\emm-1-i}\mod(x^{\emm}-1).$$
Thus, ${\bf u}(x)\circ{\bf v}(x)=0$ if and only if $S_i=0$ for $0\leq i\leq \emm-1.$
\end{proof}

\begin{lemm}\label{Lemma1}
Let ${\bf u}(x)=(u(x)\mid u'(x))$ and ${\bf v}(x)=(v(x)\mid v'(x))$ be elements in $R_{r,s}$ such that \mbox{${\bf u}(x)\circ{\bf v}(x)=0\mod(x^{\emm}-1)$}. If $u'(x)$ or $v'(x)$ equal $0$, then $u(x)v^*(x)=0\mod(x^r-1)$. Respectively,  if $u(x)$ or $v(x)$ equal $0$, then $u'(x)v'^*(x)=0\mod(x^s-1)$.
\end{lemm}

\begin{proof}
Let $u'(x)$ or $v'(x)$ equal $0$, then 
$${\bf u}(x)\circ{\bf v}(x)=u(x)\theta_{\frac{\emm}{r}}(x^r)x^{\emm-1-\deg(v(x))}v^*(x)+0=0\mod(x^{\emm}-1).$$
So, 
$$u(x)\theta_{\frac{\emm}{r}}(x^r)x^{\emm-1-\deg(v(x))}v^*(x)=\mu'(x)(x^\emm-1),$$
for some $\mu'(x)\in\ZZ_2[x]$. Let $\mu(x)=\mu'(x)x^{\deg(v(x))+1}$, by Proposition \ref{xnm=xntheta},
$$u(x)x^{\emm}v^*(x)=\mu(x)(x^r-1),$$
$$u(x)v^*(x)=0\mod(x^r-1).$$
The same argument can be used to prove the other case.
\end{proof}

\begin{prop}\label{Dimension_subcodes}
Let $C=\langle(b(x)\mid { 0}),(\ell(x)\mid  a(x))\rangle$ be a $\mathbb{Z}_2$-double cyclic code. Then,
$$|C_r|=2^{r-\deg(b(x))+\kappa}, |C_s|=2^{s-\deg(a(x))},$$
$$|(C_r)^\perp|=2^{\deg(b(x))-\kappa}, |(C_s)^\perp|=2^{\deg(a(x))},$$
$$|(C^\perp)_r|=2^{\deg(b(x))}, |(C^\perp)_s|=2^{\deg(a(x))+\kappa},$$
where $\kappa=\deg(b(x))-\deg(\gcd(b(x),\ell(x)))$.
\end{prop}

\begin{coro}\label{deg_bar_b}
Let $C=\langle(b(x)\mid { 0}),(\ell(x)\mid  a(x))\rangle$ be a $\mathbb{Z}_2$-double cyclic code with dual code $C^\perp = \langle (\bar{b}(x)\mid { 0}), (\bar{\ell}(x) \mid  \bar{a}(x)) \rangle$. Then,
\begin{eqnarray*}
\deg(\bar{b}(x)) &=&  r - \deg(\gcd(b(x),\ell(x))).
\end{eqnarray*} 
\end{coro}
\begin{proof}
It is easy to prove that $(C_r)^\perp$ is a cyclic code generated by $\bar{b}(x)$, so $|(C_r)^\perp|=2^{r-\deg(\bar{b}(x))}$. Moreover, by Proposition \ref{Dimension_subcodes}, $|(C_r)^\perp|=2^{\deg(b(x))-\kappa}$.\\
Thus, $\deg(\bar{b}(x)) =  r-\deg(\gcd(b(x),\ell(x))).$
\end{proof}

\begin{coro}\label{deg_bar_a}
Let $C=\langle(b(x)\mid { 0}),(\ell(x)\mid  a(x))\rangle$ be a $\mathbb{Z}_2$-double cyclic code with dual code $C^\perp = \langle (\bar{b}(x)\mid { 0}), (\bar{\ell}(x) \mid  \bar{a}(x)) \rangle$. Then,
\begin{eqnarray*}
\deg(\bar{a}(x)) &=&  s-\deg(a(x))-\deg(b(x))+\deg(\gcd(b(x),\ell(x))).
\end{eqnarray*} 
\end{coro}
\begin{proof}
Since $C^\perp$ is a $\ZZ_2$-double cyclic code, $(C^\perp)_s$ is a cyclic code generated by $\bar{a}(x)$, so $| (C^\perp)_s|=2^{s-\deg(\bar{a}(x))}$. Moreover, by Proposition \ref{Dimension_subcodes}, $|(C^\perp)_s|=2^{\deg(a(x))+\kappa}$.\\
Thus, $\deg(\bar{a}(x)) =  s-\deg(a(x))-\deg(b(x))+\deg(\gcd(b(x),\ell(x))).$
\end{proof}

\begin{prop}\label{CsDualInCDual}
Let $C=\langle (b(x)\mid { 0}), (\ell(x) \mid  a(x)) \rangle$ be a $\mathbb{Z}_2$-double cyclic code. Then, $\langle ({ 0} \mid  \frac{x^s-1}{a^*(x)}) \rangle\subseteq C^\perp.$
\end{prop}

\begin{proof}

Since $C_s$ is a binary cyclic code generated by $\langle a(x) \rangle$, then $(C_s)^\perp = \langle  \frac{x^s-1}{a^*(x)}\rangle$.
Let ${\bf v}(x)=(v(x)\mid v'(x))\in C$. Then, $v'(x)\in C_s$ and $({ 0}\mid \frac{x^s-1}{a^*(x)})\circ {\bf v}(x) = 0 \mod (x^{\emm}-1).$ Thus, $\langle ({ 0} \mid  \frac{x^s-1}{a^*(x)}) \rangle\subseteq C^\perp.$
\end{proof}

\begin{coro}
Let $C=\langle (b(x)\mid { 0}), (\ell(x) \mid  a(x)) \rangle$ be a $\mathbb{Z}_2$-double cyclic code with $C^\perp = \langle (\bar{b}(x)\mid { 0}), (\bar{\ell}(x) \mid  \bar{a}(x)) \rangle.$ Then, $\bar{a}(x)$ divides $\frac{x^s-1}{a^*(x)}.$
\end{coro}

\begin{coro}\label{0CsDualGen}
Let $C=\langle (b(x)\mid { 0}), (\ell(x) \mid  a(x)) \rangle$ be a $\mathbb{Z}_2$-double cyclic code. Let $T=\{({ 0}\mid p(x))\in C^\perp \} $. Then, $T$ is generated by $\langle ({ 0} \mid  \frac{x^s-1}{a^*(x)}) \rangle.$
\end{coro}

\begin{proof}
Let $T=\{({ 0}\mid p(x))\in C^\perp \}$. By Proposition \ref{CsDualInCDual}, we have that $\langle({ 0} \mid  \frac{x^s-1}{a^*(x)})\rangle\subseteq T$.\\
Since $T_s\subseteq (C_s)^\perp =\langle\frac{x^s-1}{a^*(x)}\rangle$, for all $({ 0}\mid p(x))\in T$ we have that $p(x)\in \langle\frac{x^s-1}{a^*(x)}\rangle$. Hence, there exists $\lambda(x)\in\ZZ_2[x]$ such that 
$p(x)=\lambda(x)\frac{x^s-1}{a^*(x)}$.
Therefore, for all $({ 0}\mid p(x))\in T$ we have that
$$({ 0}\mid p(x))=({ 0}\mid \lambda(x)\frac{x^s-1}{a^*(x)})=\lambda(x)\star({ 0}\mid \frac{x^s-1}{a^*(x)}).$$ 
So, $T\subseteq \langle({ 0} \mid  \frac{x^s-1}{a^*(x)})\rangle.$
\end{proof}

The previous propositions and corollaries will be helpful to determine the relations between the generator polynomials of a $\mathbb{Z}_2$-double cyclic code and the generator polynomials of its dual code. 
\begin{prop}\label{bar_b}
Let $C=\langle(b(x)\mid { 0}),(\ell(x)\mid  a(x))\rangle$ be a $\mathbb{Z}_2$-double cyclic code and $C^\perp = \langle (\bar{b}(x)\mid { 0}), (\bar{\ell}(x) \mid  \bar{a}(x)) \rangle$. Then,
$$\bar{b}(x)=\frac{x^r-1}{\gcd(b(x),\ell(x))^*}.$$
\end{prop}

\begin{proof}
We have that $(\bar{b}(x)\mid { 0})$ belongs to $C^\perp$. Then,
\begin{align*}
(b(x)\mid { 0})\circ (\bar{b}(x)\mid { 0}) = & 0 \mod(x^{\emm}-1),\\
(\ell(x)\mid a(x))\circ (\bar{b}(x)\mid { 0}) = & 0 \mod(x^{\emm}-1).
\end{align*}
Therefore, by Lemma \ref{Lemma1},
\begin{align*}
b(x)\bar{b}^*(x)= & 0 \mod(x^r-1),\\
\ell (x)\bar{b}^*(x)= & 0 \mod(x^r-1).
\end{align*}
So, $ \gcd(b(x),\ell(x))\bar{b}^*(x)= 0 \mod(x^r-1)$, and there exist $\mu(x)\in \ZZ_2[x]$ such that $ \gcd(b(x),\ell(x))\bar{b}^*(x)= \mu(x)(x^r-1)$. \\
Moreover, since $\gcd(b(x),\ell(x))$ and $\bar{b}^*(x)$ divides $(x^r-1)$, by Corollary \ref{deg_bar_b}, we have that $\deg(\bar{b}(x)) =  r-\deg(\gcd(b(x),\ell(x))).$ Then,
$$\bar{b}^*(x)=\frac{x^r-1}{\gcd(b(x),\ell(x))}.$$
\end{proof}

\begin{prop}\label{bar_a}
Let $C=\langle(b(x)\mid { 0}),(\ell(x)\mid  a(x))\rangle$ be a $\mathbb{Z}_2$-double cyclic code and $C^\perp = \langle (\bar{b}(x)\mid { 0}), (\bar{\ell}(x) \mid  \bar{a}(x)) \rangle$. Then,
$$\bar{a}(x)=\frac{(x^s-1)\gcd(b(x),\ell(x))^*}{a^*(x)b^*(x)}.$$
\end{prop}

\begin{proof}
Consider the codeword $$\frac{b(x)}{\gcd(b(x),\ell(x))}\star(\ell(x) \mid  a(x))-\frac{\ell(x)}{\gcd(b(x),\ell(x))}\star(b(x)\mid { 0})=({ 0}\mid \frac{b(x)}{\gcd(b(x),\ell(x))}a(x)).$$
Then,
\begin{align*}
(\bar{\ell}(x)\mid \bar{a}(x))\circ ({ 0}\mid \frac{b(x)}{\gcd(b(x),\ell(x))}a(x)) = 0 \mod(x^{\emm}-1).
\end{align*}
Thus, by Lemma \ref{Lemma1}
$$\bar{a}(x) \frac{a^*(x)b^*(x)}{\gcd(b(x),\ell(x))^*}  =  0 \mod(x^s-1),$$
and
$$ \bar{a}(x)\frac{a^*(x)b^*(x)}{\gcd(b(x),\ell(x))^*}=(x^s-1)\mu(x),$$
for some $\mu(x)\in\ZZ_2[x]$. It is known that $\bar{a}(x)$ is a divisor of $x^s-1$ and, by Corollary \ref{bdiviXSgcd}, we have that $\frac{a^*(x)b^*(x)}{\gcd(b(x),\ell(x))^*}$ divides $(x^s-1)$. By Corollary \ref{deg_bar_a}, $\deg(\bar{a}(x))=s-\deg(a(x))-\deg(b(x))+\deg(\gcd(b(x),\ell(x)))$, so 
$$s=\deg\left(\bar{a}(x)\frac{a^*(x)b^*(x)}{\gcd(b(x),\ell(x))^*}\right)=\deg((x^s-1)).$$
Hence, we obtain that $\mu(x)=1$ and
$$ \bar{a}(x)=\frac{(x^s-1)\gcd(b(x),\ell(x))^*}{a^*(x)b^*(x)}.$$
\end{proof}

\begin{prop}
Let $C=\langle(b(x)\mid { 0}),(\ell(x)\mid  a(x))\rangle$ be a separable $\mathbb{Z}_2$-double cyclic code. Then, $C^\perp=\langle (\frac{x^r-1}{b^*(x)}\mid { 0}), ({ 0} \mid  \frac{x^s-1}{a^*(x)}) \rangle$.
\end{prop}

\begin{coro}
Let $C$ be a separable $\mathbb{Z}_2$-double cyclic code. Then, $C^\perp$ is a separable $\mathbb{Z}_2$-double cyclic code.
\end{coro}

\begin{prop}\label{bar_ell}
Let $C=\langle(b(x)\mid { 0}),(\ell(x)\mid  a(x))\rangle$ be a non separable $\mathbb{Z}_2$-double cyclic code and $C^\perp = \langle (\bar{b}(x)\mid { 0}), (\bar{\ell}(x) \mid  \bar{a}(x)) \rangle$. Then,
$$\bar{\ell}(x)=\frac{x^r-1}{b^*(x)}\lambda(x),$$
for some $\lambda(x)\in \mathbb{Z}_2[x].$
\end{prop}

\begin{proof}
Let $\bar{c}\in C^\perp$ with $\bar{c}(x)=(\bar{b}(x)\mid 0)+(\bar{\ell}(x)\mid \bar{a}(x))$. Then
\begin{align*}
\bar{c}(x)\circ (b(x)\mid { 0}) = &  ((\bar{b}(x)\mid { 0}))\circ(b(x)\mid { 0}) + ((\bar{\ell}(x)\mid \bar{a}(x)))\circ(b(x)\mid { 0})\\
  = & 0 + ((\bar{\ell}(x)\mid \bar{a}(x)))\circ(b(x)\mid { 0})\\
  = & 0 \mod(x^{\emm}-1).
\end{align*}
So, by Lemma \ref{Lemma1}

$$\bar{\ell}(x)b^*(x)=0\mod(x^r-1)$$
and

$$\bar{\ell}(x)=\frac{x^r-1}{b^*(x)}\lambda(x).$$
\end{proof}

\begin{coro}
Let $C=\langle(b(x)\mid { 0}),(\ell(x)\mid  a(x))\rangle$ be a non separable $\mathbb{Z}_2$-double cyclic code. Then, $\deg(\lambda(x))<\deg(b(x))-\deg(\gcd(b(x),\ell(x))).$
\end{coro}

\begin{prop}\label{condition_lambda}
Let $C=\langle(b(x)\mid { 0}),(\ell(x)\mid  a(x))\rangle$ be a non separable $\mathbb{Z}_2$-double cyclic code and $C^\perp = \langle (\bar{b}(x)\mid { 0}), (\bar{\ell}(x) \mid  \bar{a}(x)) \rangle$. Let $\rho(x)=\frac{\ell(x)}{\gcd(b(x),\ell(x))}$ and $\bar{\ell}(x)=\frac{x^r-1}{b^*(x)}\lambda(x).$ Then,
{\small $$\frac{(x^\emm-1)\gcd^*(b(x),\ell(x))}{b^*(x)}\left( \lambda(x)x^{\emm-\deg(\ell(x))-1}\rho^*(x)+x^{\emm-\deg(a(x))-1}\right)=0\mod (x^\emm-1).$$}
Thus,
$$\left( \lambda(x)x^{\emm-\deg(\ell(x))-1}\rho^*(x)+x^{\emm-\deg(a(x))-1}\right)=0\mod \left(\frac{b^*(x)}{\gcd^*(b(x),\ell(x))}\right).$$
\end{prop}
\begin{proof}
Let $\rho(x)=\frac{\ell(x)}{\gcd(b(x),\ell(x))}$. Computing $(\bar{\ell}(x) \mid  \bar{a}(x))\circ(\ell(x)\mid  a(x))$ and arranging properly we obtain that 
{\small $$\frac{(x^\emm-1)\gcd^*(b(x),\ell(x))}{b^*(x)}\left( \lambda(x)x^{\emm-\deg(\ell(x))-1}\rho^*(x)+x^{\emm-\deg(a(x))-1}\right),$$}
that is equal $0\mod (x^\emm-1)$. Then,
\begin{equation}\label{equ1}
\left( \lambda(x)x^{\emm-\deg(\ell(x))-1}\rho^*(x)+x^{\emm-\deg(a(x))-1}\right)=0\mod (x^\emm-1),
\end{equation}
or
{\small
\begin{equation}\label{equ2}
 \left( \lambda(x)x^{\emm-\deg(\ell(x))-1}\rho^*(x)+x^{\emm-\deg(a(x))-1}\right)=0\mod \left(\frac{b^*(x)}{\gcd^*(b(x),\ell(x))}\right).
\end{equation}
}
Since $\frac{b^*(x)}{\gcd^*(b(x),\ell(x))}$ divides $x^\emm - 1$, clearly (\ref{equ1}) implies (\ref{equ2}).
\end{proof}
\begin{coro}\label{lambda_ell}
Let $C=\langle(b(x)\mid { 0}),(\ell(x)\mid  a(x))\rangle$ be a non separable $\mathbb{Z}_2$-double cyclic code and $C^\perp = \langle (\bar{b}(x)\mid { 0}), (\bar{\ell}(x) \mid  \bar{a}(x)) \rangle$. Let $\rho(x)=\frac{\ell(x)}{\gcd(b(x),\ell(x))}$ and $\bar{\ell}(x)=\frac{x^r-1}{b^*(x)}\lambda(x).$ Then,
$$\lambda(x)=x^{\emm-\deg(a(x))+\deg(\ell(x))}(\rho^*(x))^{-1}\mod\left(\frac{b^*(x)}{\gcd^*(b(x),\ell(x))}\right).$$
\end{coro}
\begin{proof}
Let $\rho(x)=\frac{\ell(x)}{\gcd(b(x),\ell(x))}$. By Proposition \ref{condition_lambda},
$$\left( \lambda(x)x^{\emm-\deg(\ell(x))-1}\rho^*(x)+x^{\emm-\deg(a(x))-1}\right)=0\mod \left(\frac{b^*(x)}{\gcd^*(b(x),\ell(x))}\right).$$
Then,
$$\lambda(x)x^{\emm}\rho^*(x) = x^{\emm-\deg(a(x))+\deg(\ell(x))}\mod \left(\frac{b^*(x)}{\gcd^*(b(x),\ell(x))}\right).$$
On the one hand, we have that $x^\emm=1\mod \left(\frac{b^*(x)}{\gcd^*(b(x),\ell(x))}\right).$ On the other hand, the great common divisor between $\rho(x)$ and $\frac{b(x)}{\gcd(b(x),\ell(x))}$ is $1$, then $\rho^*(x)$ is an invertible element modulo $\left(\frac{b^*(x)}{\gcd^*(b(x),\ell(x))}\right)$. Thus,
$$\lambda(x)=x^{\emm-\deg(a(x))+\deg(\ell(x))}(\rho^*(x))^{-1}\mod\left(\frac{b^*(x)}{\gcd^*(b(x),\ell(x))}\right).$$
\end{proof}

We summarize the previous results in the next theorem.

\begin{theo}
Let $C=\langle(b(x)\mid { 0}),(\ell(x)\mid  a(x))\rangle$ be a $\mathbb{Z}_2$-double cyclic code and $C^\perp = \langle (\bar{b}(x)\mid { 0}), (\bar{\ell}(x) \mid  \bar{a}(x)) \rangle$. Let $\rho(x)=\frac{\ell(x)}{\gcd(b(x),\ell(x))}$ and $\bar{\ell}(x)=\frac{x^r-1}{b^*(x)}\lambda(x).$ Then,
\begin{enumerate}
\item $\bar{b}(x)=\frac{x^r-1}{\gcd(b(x),\ell(x))^*},$
\item $\bar{a}(x)=\frac{(x^s-1)\gcd(b(x),\ell(x))^*}{a^*(x)b^*(x)},$
\item $\bar{\ell}(x)=\frac{x^r-1}{b^*(x)}\lambda(x)$, where
$$\lambda(x)x^{\emm}\rho^*(x) = x^{\emm-\deg(a(x))+\deg(\ell(x))}\mod \left(\frac{b^*(x)}{\gcd^*(b(x),\ell(x))}\right).$$
\end{enumerate}
\end{theo}

\end{document}